\documentclass[11pt,english]{article}

\usepackage[letterpaper, left=1in, right=1in, top=0.9in, bottom=0.9in]{geometry}
\usepackage[american]{babel}
\usepackage[normalem]{ulem}
\usepackage{amsmath, amssymb, cases, amsthm}
\usepackage{thmtools}
\usepackage[shortlabels]{enumitem}
\usepackage{mdframed}
\usepackage{bbm}
\usepackage{bm}
\usepackage{microtype}
\usepackage{xcolor}
\usepackage{makecell}
\usepackage{mathtools}
\usepackage{algorithmic}
\usepackage[procnumbered,ruled,vlined]{algorithm2e}
\usepackage{float}
\usepackage{varwidth}
\usepackage{tcolorbox}
\newtcolorbox{construction}[2][]
{
	colframe = gray!50,
	colback  = gray!10,
	coltitle = gray!10!black,
	left*=0mm, 
	before skip = 10pt,
	after skip = 10pt,
	title    = \textbf{\space\space #2},
	#1,
}

\SetKwInput{KwData}{Input}
\SetKwInput{KwResult}{Output}
\SetKwInput{KwGlobalVar}{Global variables}

\usepackage{nameref}
\definecolor{ForestGreen}{rgb}{0.1333,0.5451,0.1333}
\definecolor{DarkRed}{rgb}{0.8,0,0}
\definecolor{Red}{rgb}{1,0,0}
\usepackage[linktocpage=true,
pagebackref=true,colorlinks,
linkcolor=DarkRed,citecolor=ForestGreen,
bookmarks,bookmarksopen,bookmarksnumbered]
{hyperref}

\usepackage[capitalize,nosort,nameinlink]{cleveref}

\declaretheorem[numberwithin=section,refname={Theorem,Theorems},Refname={Theorem,Theorems}]{theorem}
\declaretheorem[numberlike=theorem,name=Theorem,refname={Theorem,Theorems},Refname={Theorem,Theorems}]{thm}

\declaretheorem[numberlike=theorem,name=Lemma,refname={Lemma,Lemmas},Refname={Lemma,Lemmas}]{lem}

\declaretheorem[numberlike=theorem,name=Proposition,refname={Proposition,Propositions},Refname={Proposition,Propositions}]{prop}

\declaretheorem[numberlike=theorem,style=definition]{conjecture}

\declaretheorem[numberlike=theorem,refname={Fact,Facts},Refname={Fact,Facts},name={Fact}]{fact}

\declaretheorem[numberlike=theorem,style=remark,refname={Example,Examples},Refname={Example,Examples},name={Example}]{example}

\def\final{0}  
\ifnum\final=0  
\newcommand{\todo}[1]{{\color{red}[{\tiny TODO: \bf #1}]\marginpar{\color{red}*}}}
\else 
\newcommand{\todo}[1]{}
\fi

\newcommand{\dist}{\mathrm{dist}}
\global\long\def\mincut{\mathrm{mincut}}

\begin{document}
\global\long\def\supp{\mathrm{supp}}

\title{Expander Decomposition with Almost Optimal Overhead}
\author{Nikhil Bansal\thanks{University of Michigan,         \texttt{bansaln@umich.edu}. Supported by NSF awards CCF-2327011 and CCF-2504995.} \and Arun Jambulapati\thanks{Independent, \texttt{jmblpati@gmail.com}. Supported by the NWO VICI grant 639.023.812 awarded to NB.}  \and Thatchaphol Saranurak\thanks{
        University of Michigan,         \texttt{thsa@umich.edu}.
        Supported by NSF Grant CCF-2238138 and a Sloan Fellowship.}
}

\maketitle

\pagenumbering{gobble}
\begin{abstract}
We present the first polynomial-time algorithm for computing a near-optimal \emph{flow}-expander decomposition. Given a graph $G$ and a parameter $\phi$, our algorithm removes at most a $\phi\log^{1+o(1)}n$ fraction of edges so that every remaining connected component is a $\phi$-\emph{flow}-expander (a stronger guarantee than being a $\phi$-\emph{cut}-expander). This achieves overhead $\log^{1+o(1)}n$, nearly matching the $\Omega(\log n)$ graph-theoretic lower bound that already holds for cut-expander decompositions, up to a $\log^{o(1)}n$ factor. Prior polynomial-time algorithms required removing $O(\phi\log^{1.5}n)$ and $O(\phi\log^{2}n)$ fractions of edges to guarantee $\phi$-cut-expander and $\phi$-flow-expander components, respectively.
\end{abstract}

\clearpage
\tableofcontents
\clearpage
\pagenumbering{arabic}

\section{Introduction}

Expander decomposition is a central structural primitive in graph algorithms --- where one removes a small fraction of edges from a given graph $G$, so that every connected component in the resulting graph has some desired expansion. This notion first appeared implicitly in property testing \cite{goldreich1998sublinear} and was defined explicitly in \cite{kannan2004clusterings}.

As many problems can be solved better on expanders,
over the last two decades, expander decomposition has been used widely in approximation algorithms \cite{chekuri2005multicommodity,chuzhoy2012polylogarithmic,chekuri2016polynomial}, fast graph algorithms \cite{spielman2004nearlylinear,kelner2014almost,saranurak2019expander,chuzhoy2020deterministic,chen2025maximum}, dynamic data structures \cite{patrascu2007planning,nanongkai2017dynamic,bernstein2020deterministic,goranci2021expander,goranci2023fully,jin2024fully,van2024almost}, parallel and distributed algorithms \cite{chang2021near,chang2020deterministic,haeupler2022hop,chen2025parallel,haeupler2025parallel}, and streaming algorithms \cite{chen2022weighted,filtser2023expander,chen2024streaming}. A lot of these works have also focussed on
speeding up constructions in various models of computation, or extending the notion of expansion in various interesting ways to obtain new applications. 

\paragraph{Edges removed vs.~Expansion.}
A basic and natural question is to understand how many edges must be removed to achieve a desired expansion $\phi$. Equivalently, given a budget on the number of edges that can be removed, what is the best achievable expansion.

Remarkably, this question is completely understood {\em existentially} --- for any graph on $m$ edges, removing $O(\phi m \log n)$ edges always suffices. Moreover, removing $\Omega(\phi m \log n)$ edges is necessary in general.
However, when restricted to {\em polynomial time} constructions of expander decompositions (which are necessary for most applications), all known results incur an additional polylogarithmic factor loss over this existential bound.

\smallskip

In this work, we give the {\em first} polynomial-time  cut-expander and flow-expander decompositions with only an $\log^{o(1)}n$ factor extra loss.
Before describing our results formally, we review the relevant definitions, describe the cut-and-recurse framework that achieves the existential bound, and explain why current polynomial time algorithms incur additional logarithmic losses.



\paragraph*{Cut and Flow Expander Decomposition.}
Let $G=(V,E)$ be an undirected graph with $n$ vertices and $m$ edges. We say that $G$ is a \emph{$\phi$-cut-expander} if for every subset $S\subseteq V$, we have
\[
  |E(S,V\setminus S)|\ge\phi\min\{\deg_{G}(S),\deg_{G}(V\setminus S)\},
\]
where $\deg_{G}(S)=\sum_{u\in S}\deg(u)$. 

\smallskip

Next, $G$ is a \emph{$\phi$-flow-expander} if every multi-commodity demand respecting $\deg_{G}$ can be routed in $G$ with congestion at most $1/\phi$. See \Cref{sec:prelim} for detailed definitions.

\smallskip

A \emph{$\phi$-(cut/flow)-expander decomposition} is specified by an edge set $C\subseteq E$ such that each connected component of $G-C$ is a $\phi$-(cut/flow)-expander.
We say that a $\phi$-(cut/flow)-expander decomposition has {\em overhead} $\gamma$
if the number of removed edges $|C|\le \gamma \phi m$. The overhead $\gamma$ is the main quality measure.

Flow expansion is stronger requirement than cut expansion: every $\phi$-flow-expander is a $\phi$-cut-expander, while every $\phi$-cut-expander is only guaranteed to be an $\Omega(\phi/\log n)$-flow expander by the well-known \emph{flow--cut gap} \cite{leighton1999multicommodity}. Many flow-based applications---including all-or-nothing flow \cite{chekuri2005multicommodity}, edge-disjoint paths \cite{chuzhoy2012polylogarithmic}, flow vertex sparsifiers \cite{chuzhoy2012vertex}, and oblivious routing \cite{racke2002minimizing}---therefore require decompositions with \emph{flow} expansion guarantees.

\paragraph{Benchmark: $\Omega(\log n)$ overhead is unavoidable.}
As mentioned above, already for \emph{cut}-expander decompositions, there is an inherent $\Omega(\log n)$ lower bound on the overhead. In particular, for the $n$-vertex hypercube, standard isoperimetric bounds imply that if $|C|\le m/2$, then some connected component of $G-C$ is not a $\Omega(1/\log n)$-cut expander \cite{alev2018graph}. Equivalently, setting $\phi=\Theta(1/\log n)$, any $\phi$-cut-expander decomposition must have overhead $\gamma=\Omega(\log n)$. Since flow expansion implies cut expansion, the same lower bound applies to flow-expander decompositions as well. Thus, $\Theta(\log n)$ is the best possible target for $\gamma$.

\paragraph{Cut-and-recurse and polynomial-time limitations.}
The most basic expander
decomposition procedure is to \emph{cut and recurse}: if $G$ is already a $\phi$-cut-expander, return $C=\emptyset$. Otherwise, there exists some \emph{$\phi$-sparse} cut $S$ satisfying
\[
  |E(S,V\setminus S)|\le\phi\min\{\deg_{G}(S),\deg_{G}(V\setminus S)\},
\]
and we recurse on $G[S]$ and $G[V\setminus S]$, returning $E(S,V\setminus S)$ and  the cuts produced recursively. 

\smallskip

So assuming that we can solve  $\phi$-sparse cut exactly, the above procedure,  together with a standard charging argument,\footnote{For each cut $(S,V\setminus S)$ charge $|E(S,V\setminus S)|$ to vertices on the smaller side so that each such vertex $v$ is charged at most $\phi\deg(v)$. Each vertex can lie on the smaller side at most $\log n$ times, yielding total charge at most $2\phi(\log n)\,m$ and hence $\gamma=O(\log n)$.}
gives overhead $\gamma=O(\log n)$.
Together with the hypercube lower bound above, this overhead is graph-theoretically optimal (for appropriate $\phi$).

\smallskip

In polynomial time however, one cannot find such sparse cuts exactly. In particular, when $G$ is not a $\phi$-cut-expander, using the best known $O(\sqrt{\log n})$-approximation to sparsest cut \cite{arora2009expander}, we can only find an $O(\phi\sqrt{\log n})$-sparse cut. For flow expansion, the obstruction is even more fundamental: when $G$ is not a $\phi$-flow-expander, there may not even exist any $o(\phi\log n)$-sparse cut due to the flow--cut gap; and thus we can only find an $O(\phi\log n)$-sparse cut \cite{leighton1999multicommodity}. Plugging these extra losses into cut-and-recurse procedure yields, for any $\phi$,
\begin{itemize}
  \item a $\phi$-cut-expander decomposition with overhead $O(\log^{1.5}n)$ in polynomial time, and
  \item a $\phi$-flow-expander decomposition with overhead $O(\log^{2}n)$ in polynomial time.
\end{itemize}

\paragraph{State of the art: the long-standing gap.}
Surprisingly, despite extensive research on expander decomposition and its extensions, no polynomial-time algorithm is known that  improves upon the overhead bounds described above. In particular, it has remained open for over two decades whether one can match the $\Omega(\log n)$ existential benchmark in polynomial time, and more strongly, whether one can do so for \emph{flow}-expansion.

\paragraph{Our result.} Even though improving upon the $O(\sqrt{\log n})$-approximation to sparsest cut is a major open problem, and the $O(\log n)$ flow-cut gap loss may seem inherently unavoidable, we are able to circumvent these natural barriers and match the $\Omega(\log n)$ benchmark up to a $\log^{o(1)}(n)$ factor. This holds even for flow expanders. In particular, we show the following.
\begin{thm}\label{thm:main intro}
  There is a polynomial-time algorithm that, given an undirected graph $G=(V,E)$ with $n$ vertices and $m$ edges and parameter $\phi$, returns an edge set $C\subseteq E$ such that each connected component of $G-C$ is a $\phi$-flow-expander and $|C|\le\phi\gamma m$, where $\gamma=O(\log(n)\exp(\sqrt{\log\log n}))$.
\end{thm}

\Cref{thm:main intro} yields the first polynomial-time expander decomposition construction with optimal overhead up to a $\log^{o(1)}n$ factor. It improves the state-of-the-art by $\log^{0.5-o(1)}n$ and $\log^{1-o(1)}n$ factors for cut-expander decompositions and flow-expander decompositions, respectively. The result also extends to capacitated graphs, the terminal version, and more general node-weighting versions; see \Cref{thm:main body} for the theorem in full generality.

\smallskip

At a high level, we adapt the recent approach of graph clustering on top of a spreading-metric LP solution in \cite{bansal2024approximating}---developed to obtain improved approximations for cutwidth---to the expander decomposition setting, enabling a strictly better structural tradeoff than cut-and-recurse in polynomial time.  A high-level intuition  of the algorithm appears in \Cref{sec:ed:overview}.

\section{Preliminaries}

\label{sec:prelim}
All logs are base $2$.
We work with an undirected graph $G=(V,E)$ where edges have unit capacity.
For a vertex set $S\subseteq V$, let $\delta_{G}(S) := E(S,V\setminus S)$ be the set of edges crossing the cut $(S,V\setminus S)$.

\paragraph{Node-weighting.} A \emph{node-weighting} is a function $A:V\to\mathbb{R}_{\ge0}$. For $S\subseteq V$ we use the shorthand
\[
  A(S):=\sum_{v\in S}A(v),\qquad |A|:=A(V)=\sum_{v\in V}A(v).
\]
We also use the restriction $A_{S}:V\to\mathbb{R}_{\ge0}$ defined by $A_{S}(v)=A(v)$ if $v\in S$ and $A_{S}(v)=0$ otherwise, so that $|A_{S}|=A(S)$. The degree node-weighting is $\deg_{G}:V\to\mathbb{Z}_{\ge0}$ where $\deg_{G}(v)$ is the degree of $v$ in $G$ and $\deg_{G}(S)=\sum_{v\in S}\deg_{G}(v)$.

\paragraph{Multi-commodity Demands.}
A \emph{demand} is a function $D:\binom{V}{2}\to\mathbb{R}_{\ge0}$, where $D(u,v)$ specifies how much flow must be sent between the (unordered) pair $\{u,v\}$. We say that $D$ \emph{respects} $A$ (or is \emph{$A$-respecting}) if for every $x\in V$,
\[
  \sum_{y\in V\setminus\{x\}} D(x,y) \le A(x).
\]
We say that a demand $D$ is \emph{routable in $G$ with congestion $\rho$} if there exists a feasible multi-commodity flow that routes each pair-demand $D(u,v)$ in $G$ such that the total flow on every edge $e\in E$ is at most $\rho$ (recall that edge capacities are $1$).
Given $A$, the \emph{$A$-product demand} is the demand $D_{A}$ defined by
\[
  D_{A}(u,v):=\frac{A(u)A(v)}{|A|}\qquad \text{for all } \{u,v\}\in\binom{V}{2}.
\]
Note that $D_{A}$ respects $A$, as $\sum_{v} D_A(u,v) \leq A(u)$ for each $u$.
\paragraph{Cut Expansion.}
We recall two notions of expansion with respect to a node-weighting $A$ and parameter $\phi>0$. The weighting $A$ is \emph{$\phi$-cut-expanding in $G$} if for every nontrivial $\emptyset\subsetneq S\subsetneq V$,
\[
  |\delta_{G}(S)| \ge \phi\,\min\{A(S),A(V\setminus S)\}.
\]
A set $S$ is a \emph{$\phi$-sparse cut (with respect to $A$)} if it violates the above inequality.

\paragraph{Flow Expansion.}
The weighting $A$ is \emph{$\phi$-flow-expanding in $G$} if every $A$-respecting demand is routable in $G$ with congestion at most $1/\phi$.
Fact \ref{fact:product-demand} below shows that, upto a factor of $2$, $A$ is $\phi$-flow-expanding in $G$ iff the $A$-product demand $D_A$ is routable in $G$ with congestion $1/\phi$. That is, $D_A$ is the ``hardest'' $A$-respecting demand. (The proof is given in \Cref{sec:omit} for completeness.)
\begin{fact}\label{fact:product-demand}
  If the $A$-product demand $D_{A}$ is routable in $G$ with congestion $1/\phi$, then $A$ is $(\phi/2)$-flow-expanding in $G$.
\end{fact}

\paragraph{Expanders.}
Finally, $G$ is a \emph{$\phi$-flow-expander} (resp., \emph{$\phi$-cut-expander}) if $\deg_{G}$ is $\phi$-flow-expanding (resp., $\phi$-cut-expanding) in $G$.

\section{Expander Decomposition Algorithm}

We now prove the following main result of the paper:
\begin{thm}\label{thm:main body}
  For any undirected graph $G=(V,E)$, an integral node-weighting $A:V\rightarrow\mathbb{Z}_{\ge0}$, and parameter $\phi$, we can efficiently compute $C\subseteq E$ such that
  \begin{itemize}
    \item $|C|\le\phi|A|\log(|A|)\exp(O(\sqrt{\log\log|A|}))$, and
    \item For each connected component $U$ in $G-C$, $A\cap U$ is $\phi$-flow-expanding in $G[U]$.
  \end{itemize}
\end{thm}

This implies \Cref{thm:main intro} by setting $A = \deg_G$.


\paragraph{Organization.}
In this section, we first describe the algorithm in \Cref{sec:ed:description} and give its high-level explanation in \Cref{sec:ed:overview}. Then, we prove that the algorithm is well-defined in \Cref{sec:ed:validity}. We formally prove the expansion guarantee and bound the cut size of the decomposition in \Cref{sec:ed:correct}.

\subsection{Algorithm Description}\label{sec:ed:description}
\Cref{alg:ED} uses two standard tools including the Concurrent Multi-Commodity Flow linear program and sparse neighborhood covers, described below.

\paragraph{Concurrent Flow Linear Program.}

Let $D$ be the $A$-product demand where $D(u,v)=\frac{A(u)A(v)}{A(V)}$ for all $u,v$. We consider the following LPs which are dual to each other:

\begin{align*}\label{lp:concurrent-mcf}
  & (\text{Primal}) &  &  &  &  &  & (\text{Dual})\\
  \min\quad & \kappa &  &  &  &  & \max\quad & \sum_{u,v}D(u,v)\text{dist}_{\ell}(u,v)\\
  \text{s.t.}\quad & \sum_{p\ni e}f_{p}\le\kappa & \forall e\in E &  &  &  & \text{s.t.}\quad & \sum_{e\in E}\ell_{e}\le1\\
  & \sum_{p:(u,v)\text{ paths}}f_{p}=D(u,v) & \forall u,v &  &  &  &  & \ell_{e}\ge0 & \forall e\in E\\
  & f_{p}\ge0 & \forall \text{path }p
\end{align*}
The primal LP is the problem of finding a multi-commodity flow that routes the $A$-product demand $D$ with minimum congestion. The dual LP is a relaxation of finding sparse cuts w.r.t. $A$ (see e.g., \cite{Vazirani2001Approximation}).
Both LPs can be solved in polynomial time (e.g., via a compact formulation for flow LP).

\paragraph{Sparse Neighborhood Cover.}
The second tool is the standard sparse neighborhood cover. We use the formulation from \cite{bansal2024approximating}.
\begin{lem}
  [Lemma 6 of \cite{bansal2024approximating}]\label{lemma:clustering}For any weighted graph $G=(V,E)$, terminal set $T=\{v_{1},\dots,v_{|T|}\}\subseteq V$, radius parameter $R$, there is an efficient algorithm $\textsc{cluster}(G,T,R)$ that returns a collection ${\cal S}=\{S_{1},\dots,S_{|T|}\}$ \emph{of disjoint} vertex sets where
  \begin{enumerate}
    \item \textbf{(Cut size):}$\sum_{S\in{\cal S}}|\delta(S)|\le O(\log|T|)\frac{\sum_{e\in E}w_{e}}{R}$,
    \item \textbf{(Covering):} $\bigcup_{S\in{\cal S}}S\supseteq\bigcup_{v\in T}B(v,R)$,
    \item \textbf{(Diameter):} $S_{i}\subseteq B(v_{i},2R)$ for each $v_{i}\in T$.
  \end{enumerate}
  \label{lem:cluster}
\end{lem}

The first condition bounds the total cut size of all clusters $S\in{\cal S}$. It is very crucial that the overhead factor is $O(\log |T|)$ and not just $O(\log n)$.  The second shows that all clusters cover $R$-radius balls around every terminal. The third implies that the weak diameter of each cluster is at most $4R$.

\medskip

Consider the following algorithm.

\medskip 

\begin{algorithm}[H]
  Define $\gamma=\exp(\sqrt{\log\log|A|})$ and $L=\left\lceil \log_{\gamma}\log|A|\right\rceil +1$
  \begin{enumerate}
    \item Solve the concurrent flow LP. If $\kappa<1/\phi$, then return. Else, let $\ell$ be the edge weights in $G$ such that
      \[
        \sum_{e\in E}\ell_{e}\le1\text{ and }\sum_{u,v}D(u,v)\text{dist}_{\ell}(u,v)\ge1/\phi.
      \]
    \item Define the \emph{radius threshold}
      \[
        \Delta_{i}:=1/(4\phi|A|8^{i})\text{ for each }i\ge0.
      \]
      Define the \emph{mass threshold}
      \[
        a_{j}:=|A|/2^{\gamma^{j}}\text{ for each }j\ge-1.
      \]
    \item \textbf{\label{step:heavy-cluster}(Heavy-cluster case):} If there exists  $x\in\supp(A)$ with $A(B(x,\Delta_{0}))\ge|A|/2$, then perform sweep cut from $A(B(x,\Delta_{0}))$ and find a $O(\phi)$-sparse cut $S'$ w.r.t. $A$. (See \Cref{sec:heavy case} for the precise algorithm) and return
      \[
        E(S',V-S')\cup\text{ED}(G[S'],A_{S'},\phi)\cup\text{ED}(G[V-S'],A_{V-S'},\phi).
      \]
    \item For each $x\in\supp(A)$, define the \emph{radius scale} of $x$ as the smallest $i_{x}\ge1$ where
      \[
        \log\frac{|A|}{A(B(x,\Delta_{i}))}\le\gamma\cdot\log\frac{|A|}{A(B(x,\Delta_{i-1}))}.
      \]
      The \emph{mass scale} $j_{x}$ is such that
      \[     A(B(x,\Delta_{i_{x}}))\in(a_{j_{x}},a_{j_{x}-1}].
      \]
      We have $i_{x}\in[1,L]$ and $j_{x}\in[1,L]$ by \Cref{prop:bound radius,prop:bound size}
    \item For each radius and mass scale $(i,j)$, let $V_{i,j}=\{x\in\supp(A)\mid i_{x}=i$ and $j_{x}=j\}$. Let
      \[
        (i^{*},j^{*})=\arg\max_{(i,j)}A(V_{i,j}).
      \]
    \item Let $N\subseteq V_{i^{*},j^{*}}$ be a \emph{net} obtained by a maximal packing of balls of radius $\Delta_{i^{*}}$ around vertices in $V_{i^{*},j^{*}}$. That is, $\{B(x,\Delta_{i^{*}})\}_{x\in N}$ are disjoint and we have $V_{i^{*},j^{*}}\subseteq\bigcup_{x\in N}B(x,2\Delta_{i^{*}})$.
    \item Compute ${\cal S}=\textsc{cluster}(G,T:=N,R:=2\Delta_{i^{*}})$.
    \item \textbf{(Balanced case):} Return
      \[
        \bigcup_{S\in{\cal S}}\delta(S)\cup\bigcup_{S\in{\cal S}}\text{ED}(G[S],A_{S},\phi)\cup\text{ED}(G[V-V({\cal S})],A_{V-V({\cal S})},\phi).
      \]
  \end{enumerate}
  \caption{$\text{ED}(G=(V,E),A,\phi)$}
  \label{alg:ED}
\end{algorithm}

\subsection{High-level Explanation of the Algorithm}\label{sec:ed:overview}
\label{sec:high-level}

This section is informal: the goal is to explain how we adapt the spreading-metric clustering idea of~\cite{bansal2024approximating} to beat the usual ``cut-and-recurse'' overhead in polynomial time.

\paragraph*{From routability to a spreading metric.}
\Cref{alg:ED} begins by solving the concurrent flow LP for the $A$-product demand
$D(u,v)=A(u)A(v)/|A|$.
If the optimum congestion $\kappa < 1/\phi$, then the $A$-product demand is routable with congestion
$1/\phi$, and hence $A$ is already $(\phi/2)$-flow-expanding by \Cref{fact:product-demand}; we can safely stop.

So the interesting case is when $\kappa \ge 1/\phi$.
By LP duality, we obtain a nonnegative length function $\ell$ on edges such that
\[
  \sum_{e\in E}\ell_e \le 1
  \qquad\text{and}\qquad
  \sum_{u,v} D(u,v)\,\dist_\ell(u,v) \;\ge\; \frac{1}{\phi}.
\]
Equivalently, if we sample a random pair $(u,v)$ according to the $A$-product distribution, then
\[
  \mathbb{E}\big[\dist_\ell(u,v)\big] \;\gtrsim\; \frac{1}{\phi|A|}.
\]
Thus, in the metric defined by the lengths $\ell$, a typical pair of $A$-mass points is far apart:
$\ell$ is a \emph{spreading metric} for the demand.

\paragraph*{Warm-up: why a single ``good scale'' would solve the problem.}
To see the key idea, imagine the following idealized situation.
Suppose there exists a radius $\Delta^\star = \Theta(1/(\phi|A|))$ and a mass scale $a^\star$ such that
for \emph{every} vertex $x$,
\begin{equation}
  \label{eq:perfect-scale}
  a^\star \;\le\; A(B(x,\Delta^\star))
  \;\le\; A(B(x,4\Delta^\star))
  \;\le\; 2a^\star,
\end{equation}
where $B(x,r)$ denotes the ball of radius $r$ around $x$ in the $\ell$-metric.
Intuitively, \eqref{eq:perfect-scale} says that at the ``right zoom level'' $\Delta^\star$,
every point sees about $a^\star$ units of $A$-mass nearby, and that increasing the radius by a constant
factor does not suddenly swallow much more $A$-mass.

Now take a maximal packing net $N$ of $\Delta^\star$-balls:
the balls $\{B(x,\Delta^\star)\}_{x\in N}$ are disjoint, but the $2\Delta^\star$-balls cover $V$.\footnote{The factor $2$ for covering the whole $V$ here is the crucial reason why our end result suffers the $\exp(\sqrt{\log \log |A|})$ factor and does not achieve $O(\log\log |A|)$. The factor $2$ is tight; there exists a metric (the projective plane) where, for every $R' < 2\Delta^\star$, the $R'$-balls can cover only a square root fraction of the metric.}
By disjointness and \eqref{eq:perfect-scale}, the size of the net is at most 
\[
  |N| \;\le\; \frac{|A|}{a^\star}.
\]
Following the clustering step of~\cite{bansal2024approximating}, run the neighborhood cover routine from \Cref{lem:cluster} with terminals $T:=N$ and radius
$R:=2\Delta^\star$, obtaining disjoint clusters $\mathcal{S}$.
By \Cref{lem:cluster} the number of cut edges satisfies
\[
  \sum_{S\in\mathcal{S}} |\delta(S)|
  \;\le\;
  O(\log|N|)\cdot \frac{\sum_{e}\ell_e}{2\Delta^\star}
  \;\lesssim\;
  \phi|A|\cdot \log\frac{|A|}{a^\star}.
\]
Moreover, by the diameter guarantee of \Cref{lem:cluster}, each cluster $S\in\mathcal{S}$ lies inside
some ball of radius $4\Delta^\star$ around its terminal, so by \eqref{eq:perfect-scale} it satisfies
$A(S)\le 2a^\star$.
Therefore the total \emph{recursive} cost on the clusters, under the standard potential
$\phi\cdot A(\cdot)\log A(\cdot)$, is at most
\[
  \sum_{S\in\mathcal{S}} \phi\,|A_S|\,\log|A_S|
  \;\le\;
  \sum_{S\in\mathcal{S}} \phi\,|A_S|\,\log(2a^\star)
  \;\lesssim\;
  \phi|A|\,\log a^\star.
\]
Putting the two contributions together yields the telescoping expression
\[
  \phi|A|\log\frac{|A|}{a^\star} \;+\; \phi|A|\log a^\star \;=\; \phi|A|\log|A|.
\]
This is the guiding principle of the whole algorithm:
\begin{quote}
  \emph{We want one clustering step whose boundary cost is $\log(\text{\#clusters})\approx \log(|A|/a)$,
  while ensuring each cluster has mass $\approx a$, so that recursion only pays $\log a$.}
\end{quote}

\paragraph*{Reality: different vertices have different good scales, so we ``vote'' over scales.}
As in~\cite{bansal2024approximating}, the main difficulty is that, in general, there is no single radius $\Delta^\star$ that satisfies \eqref{eq:perfect-scale} for all $x$.
Instead, \Cref{alg:ED} discretizes radii and masses into $L$ scales: for $0\le i \le L$ and $-1 \le j \le L$,
\begin{align}
  \Delta_i := \frac{1}{4\phi|A|\,8^i},
  \qquad
  a_j := \frac{|A|}{2^{\gamma^j}},
  \label{eq:basic-recurrences}
\end{align}
where $\gamma=\exp(\sqrt{\log\log|A|})$ and $L=\left\lceil \log_{\gamma}\log|A|\right\rceil +1= O(\sqrt{\log\log|A|})$.
\footnote{
These parameters are chosen to optimize the various tradeoffs that arise in the analysis later.}

\medskip

For a fixed vertex $x$, consider the growth curve of ball masses as we zoom in:
\[
  A(B(x,\Delta_0)),\ A(B(x,\Delta_1)),\ \ldots,\ A(B(x,\Delta_L)).
\]
Since $A$ is integral, $A(B(x,\Delta_L))\ge 1$ always.
On the other hand, if already $A(B(x,\Delta_0))\ge |A|/2$, then $x$ lies inside a very large
``heavy'' ball, which is handled separately in Step~\ref{step:heavy-cluster} (we discuss this case below).
So in the balanced regime we may assume $A(B(x,\Delta_0))\le |A|/2$ for all $x\in\supp(A)$.

Now define the \emph{radius scale} $i_x$ of $x$ to be the \emph{first} index where the ball stops
shrinking ``too fast'' in a logarithmic sense:
\[
  \log\!\Big(\frac{|A|}{A(B(x,\Delta_{i_x}))}\Big)
  \;\le\;
  \gamma \cdot \log\!\Big(\frac{|A|}{A(B(x,\Delta_{i_x-1}))}\Big).
\]
Intuitively, we keep zooming in until the ``difficulty measure''
$\log(|A|/A(B(x,\Delta_i)))$ fails to increase by a factor $\gamma$ in one step.
Such an index must exist (and indeed $i_x\le L$), because if the measure increased by a factor
$\gamma$ at \emph{every} step, then after $L$ steps it would exceed $\log|A|$, forcing
$A(B(x,\Delta_L))<1$, which is impossible.

\begin{example}
Suppose that as we zoom in by a factor $8$ in radius, the $A$-mass near $x$ shrinks like
\[
  A(B(x,\Delta_0))\approx |A|/2,\qquad A(B(x,\Delta_1))\approx |A|/2^{\gamma^{10}},\qquad A(B(x,\Delta_2))\approx |A|/2^{\gamma^{11}}.
\]
Then $\log(|A|/A(B(x,\Delta_1)))$ is much larger than $\log(|A|/A(B(x,\Delta_0)))$ by a $\gamma^{10}$ factor, but the increase from $\Delta_1$ to $\Delta_2$ is mild (only a $\gamma$ factor). Then, $i_x =2 $ is the first such ``stabilization'' scale.
\end{example}

Next define the \emph{mass scale} $j_x$ by bucketing the mass at the stabilization radius:
\[
  A(B(x,\Delta_{i_x})) \in (a_{j_x},\,a_{j_x-1}].
\]
Observe that indeed $1 \le j_x \le L$ as calculated in \Cref{prop:bound size}.

So each $x\in\supp(A)$ chooses one pair $(i_x,j_x)\in[L]\times[L]$.
We now let the vertices ``vote'' for their chosen pair:
for each $(i,j)$, define $V_{i,j}:=\{x\in\supp(A)\mid i_x=i,\ j_x=j\}$, and pick
\[
  (i^\ast,j^\ast) := \arg\max_{(i,j)} A(V_{i,j}).
\]
Since there are only $L^2$ choices, this guarantees a large consensus class
\begin{equation}
  \label{eq:consensus}
  A(V_{i^\ast,j^\ast}) \;\ge\; \frac{|A|}{L^2}.
\end{equation}
You should think of $V_{i^\ast,j^\ast}$ as a large set of vertices that agree on
(1) the same ``good zoom level'' $\Delta_{i^\ast}$, and (2) the same local mass scale $\approx a_{j^\ast}$.

\paragraph*{The clustering step at the consensus scale.}
We now rerun the warm-up argument, but only using terminals from the consensus class.
Let $N\subseteq V_{i^\ast,j^\ast}$ be a maximal packing net of $\Delta_{i^\ast}$-balls.
Disjointness and the definition of $j^\ast$ imply
\[
  |N| \le |A|/a_{j^\ast}.
\]
We call $\textsc{cluster}(G,T:=N,R:=2\Delta_{i^\ast})$ and obtain disjoint clusters $\mathcal{S}$.

Two properties make this step useful.

\smallskip
\noindent\emph{(i) We cut around a nontrivial amount of $A$-mass.}
Because $N$ is maximal, the $2\Delta_{i^\ast}$-balls around $N$ cover $V_{i^\ast,j^\ast}$, and the
covering guarantee of \Cref{lem:cluster} then implies that the union of clusters $V(\mathcal{S})$ contains all of
$V_{i^\ast,j^\ast}$.
By \eqref{eq:consensus}, this means the recursion peels off at least $|A|/L^2$ mass in one shot.

\smallskip
\noindent\emph{(ii) Each cluster is ``small'' in $A$-mass.}
By the diameter guarantee of \Cref{lem:cluster}, each cluster $S\in\mathcal{S}$ lies inside a ball of radius
$4\Delta_{i^\ast}$ around its terminal $v\in N$.
Since $v$ has radius scale $i^\ast$, the mass of a $(\Theta(\Delta_{i^\ast}))$-ball around $v$
cannot jump too much when we expand the radius by a constant factor; combined with the fact that
$A(B(v,\Delta_{i^\ast}))\approx a_{j^\ast}$ (by the definition of $j^\ast$), this yields an upper bound
of the form
\[
  A(S) \;\lesssim\; a_{j^\ast-O(1)}.
\]
(\Cref{sec:ed:correct} makes this precise; the point here is that the stabilization rule defining $i_x$ is exactly
  what prevents a cluster of weak diameter $O(\Delta_{i^\ast})$ from capturing much more than
the local scale $a_{j^\ast}$.)

\smallskip
\noindent\emph{Boundary accounting.}
Finally, \Cref{lem:cluster} bounds the total boundary of the clusters by
\[
  \sum_{S\in\mathcal{S}}|\delta(S)|
  \;\le\;
  O(\log|N|)\cdot \frac{\sum_e \ell_e}{2\Delta_{i^\ast}}
  \;\approx\;
  \phi|A|\cdot \log|N|
  \le
  \phi|A|\cdot \log\frac{|A|}{a_{j^\ast}}.
\]
This is exactly the ``$\log(\#\text{clusters})$'' term from the warm-up.
Meanwhile, since each cluster has $A$-mass at most $\approx a_{j^\ast-O(1)}$,
the recursion on these clusters only pays $\log a_{j^\ast}$ in the potential.
Thus the same telescoping intuition from the ideal case survives:
\[
  \text{(boundary)}\;\approx\;\log\frac{|A|}{a}
  \qquad\text{and}\qquad
  \text{(recursion)}\;\approx\;\log a
  \quad\Longrightarrow\quad
  \text{total}\;\approx\;\log|A|.
\]
The technical work in \Cref{sec:ed:validity,sec:ed:correct} is to make the words ``$\approx$'' precise and to keep the
losses within $\log^{o(1)}|A|$ factors (coming from the $L$ and $\gamma$ discretization).

\paragraph*{Why we need the heavy-cluster case.}
The only way the ``zoom until stabilization'' story can fail is if we start with a vertex $x$ whose
coarsest ball already contains a constant fraction of the total mass,
$A(B(x,\Delta_0))\ge |A|/2$.
In this situation there is a dense core $K:=B(x,\Delta_0)$ of small $\ell$-diameter containing a large
fraction of $A$.
The dual constraint $\sum_{u,v}D(u,v)\dist_\ell(u,v)\ge 1/\phi$ then forces a significant amount of
demand to cross any sweep cut that moves away from $K$.
A Leighton-Rao sweep algorithm (formalized in \Cref{sec:heavy case}) finds a cut $S'\supseteq K$ that is
$O(\phi)$-sparse with respect to $A$. Since we only pay $O(1)$ overhead factor in this level, we can revert to the standard cut-and-recurse step on
$G[S']$ and $G[V\setminus S']$.
This is exactly Step~\ref{step:heavy-cluster} of \Cref{alg:ED}.

\paragraph*{Why Seymour's telescoping trick does not directly apply.}
Readers familiar with telescoping-volume arguments may wonder whether one can do even better,
for example by charging each level by $\log(\mathrm{vol}_{\ell}(\text{parent})/\mathrm{vol}_{\ell}(\text{child}))$.
The obstruction is that our metric $\ell$ is \emph{not fixed across recursion}.
When we recurse on an induced subgraph $G[S]$, we re-solve the flow LP inside $G[S]$, producing a
new dual metric $\ell_S$ that can be essentially unrelated to the restriction of $\ell$.
Therefore, any ``$\ell$-volume'' potential need not be consistent: a set that looks small in the parent
metric can become large again under the new metric, preventing a clean telescoping argument.

Our approach avoids this instability by telescoping with a quantity that is invariant under recursion:
the true $A$-mass.
We explicitly choose clusters so that the boundary cost depends on $\log(|A|/a)$ (how many clusters we
create), while the recursive cost depends on $\log a$ (how large clusters are in true mass).
Because $A(\cdot)$ does not change when we recurse, these logs add up cleanly to $\log|A|$.

\subsection{Validity of the Algorithm}\label{sec:ed:validity}\label{sec:ed:validity}

First, we show that indeed the radius scale of each $x$ is $i_{x}<L$. The idea is, otherwise, $A(B(x,\Delta_{i+1}))$ would shrink rapidly compared to $A(B(x,\Delta_{i}))$. But even the smallest ball $B(x,\Delta_{L})$ has $A(B(x,\Delta_{L}))\ge1$. So, this can happen only when $A(B(x,\Delta_{0}))>|A|/2$. But this contradicts that there is no heavy cluster from Step \ref{step:heavy-cluster}.
\begin{prop}
  \label{prop:bound radius}For any $x\in\supp(A)$, $1\le i_{x}\le L$.
\end{prop}

\begin{proof}
  Suppose that $i_{x}>L$. Then, for all $1\le i\le L$, $\log\frac{|A|}{A(B(x,\Delta_{i}))}>\gamma\cdot\log\frac{|A|}{A(B(x,\Delta_{i-1}))}.$ Thus,
  \[
    \log\frac{|A|}{A(B(x,\Delta_{L}))}\ge\gamma^{L}\log\frac{|A|}{A(B(x,\Delta_{0}))}.
  \]
  We have $A(B(x,\Delta_{L}))\ge1$ as $A$ is integral and $A(B(x,\Delta_{0}))\le|A|/2$, otherwise the algorithm would be in the heavy-cluster case. So,

  \[
    \log|A|\ge\log\frac{|A|}{A(B(x,\Delta_{L}))}\text{ and }\log\frac{|A|}{A(B(x,\Delta_{0}))}\ge1
  \]
  and, hence,
  \[
    \log|A|\ge\gamma^{L}.
  \]
  But, the choice of $\gamma$ and $L$, we have $\gamma^{L}>\log|A|$ which is a contradiction.
\end{proof}
\begin{prop}
  \label{prop:bound size}For any $x\in\supp(A)$, $1\le j_{x}\le L$.
\end{prop}

\begin{proof}
  We have
  \[
    a_{L}<1\le A(B(x,\Delta_{L-1}))\le A(B(x,\Delta_{0}))\le|A|/2=a_{0}
  \]
  because, again, $A$ is integral and there is no heavy cluster in Step \ref{step:heavy-cluster}. Since we assign the mass scale $j_{x}$ to $x$ if $A(B(x,\Delta_{i_{x}}))\in(a_{j_{x}},a_{j_{x}-1}]$, the claim follows.
\end{proof}

\subsection{Correctness}\label{sec:ed:correct}

We verify that the expansion guarantee in each component of $G-C$.
\begin{prop}
  For each connected component $U$ in $G-C$, $A\cap U$ is $(\phi/2)$-flow-expanding in $G[U]$.
\end{prop}

\begin{proof}
  When the optimal solution of the LP is $\kappa<1/\phi$. This means that the $A$-product demand is routable in $G$ with congestion less than $1/\phi$. So, $A$ is a $\phi/2$-flow-expanding in $G$. The claim follows from applying this argument on each induced subgraph in the recursion.
\end{proof}
The next two lemmas bound the cut size (excluding the recursion). The heavy-cluster case is easy and we obtain an $O(1)$-approximate sparsest cut by closely following the technique by Leighton and Rao \cite{leighton1999multicommodity}. Thus, we defer this standard proof to \Cref{sec:heavy case}.
\begin{lem}
  \label{lem:heavy cut}$|\delta(S')|\le12\phi\min\{A(S'),A(V-S')\}$ and $S'$ can be found efficiently.
\end{lem}

In the balanced case, we bound the cut size as follows.

\begin{lem}
  \label{lem:cluster cut}$\sum_{S\in{\cal S}}|\delta(S)|\le c_{0}\phi8^{L}\gamma^{2}\cdot|A|\log\frac{|A|}{a_{j^{*}-2}}$ for some constant $c_{0}$.
\end{lem}

\begin{proof}
  For each $x\in N$, we have $A(B(x,\Delta_{i^{*}}))>a_{j^{*}}$ since $x$ has radius scale $i^{*}$ and mass scale $j^{*}$. But the $\Delta_{i^{*}}$-radius balls around each $x\in N$ are disjoint. So
    $|N|\le|A|/a_{j^{*}}$.
  As from definition $\log(|A|/a_{j})=\gamma^{j}$ for all $j$, we have the bound
  \begin{align}
    \log|N|\le\log\frac{|A|}{a_{j^{*}}}=\gamma^{2}\log\frac{|A|}{a_{j^{*}-2}}.
    \label{eq:log-net-bound}
  \end{align}
  By the cut size guarantee of $\textsc{cluster}$ from \Cref{lemma:clustering}, we have that
  \begin{align*}
    \sum_{S\in{\cal S}_{(i^{*},j^{*})}}|\delta(S)| & =O(\log|N|)\frac{\sum_{e\in E}\ell_{e}}{2\Delta_{i^{*}}} \\ 
    &\le c_{0}\phi8^{L}\gamma^{2}\cdot|A|\log\frac{|A|}{a_{j^{*}-2}}
  \end{align*}
  for some large enough constant $c_{0}$. Here, in the second line we used that $\sum_{e\in E}\ell_{e}\le1$,  the bound in \eqref{eq:log-net-bound} and that $\Delta_{i^{*}}=1/(4\phi|A|8^{i^{*}}) \geq  1/(4\phi|A|8^{L})$.
\end{proof}
The next two lemmas are needed for our the induction proof.
\begin{lem}
  \label{lem:bound size cluster}For each cluster $S\in{\cal S}$, $|A_{S}|\le a_{j^{*}-2}$
\end{lem}

\begin{proof}
  By the diameter bound of $\textsc{cluster}$ from \Cref{lemma:clustering}, $S\subseteq B(x,4\Delta_{i^{*}})\subseteq B(x,\Delta_{i^{*}-1})$ for some $x\in N$. By the definition of radius scale, we have
  \[
    \frac{|A|}{A(B(x,\Delta_{i^{*}}))}\le \left(\frac{|A|}{A(B(x,\Delta_{i^{*}-1}))}\right)^{\gamma}.
  \]
  Thus,
  \begin{align*}
    A(B(x,\Delta_{i^{*}-1})) & \le|A|^{1-1/\gamma}A(B(x,\Delta_{i^{*}}))^{1/\gamma}\\
    & \le|A|^{1-1/\gamma}a_{j^{*}-1}^{1/\gamma}\\
    & =a_{j^{*}-2}.
  \end{align*}
  where the second line is because $x$ has mass scale $j^{*}$. To see the last line, write $a=|A|$. Observe that
  \[
    |A|^{1-1/\gamma}a_{j^{*}-1}^{1/\gamma}=a^{1-1/\gamma}\frac{a^{1/\gamma}}{2^{(\gamma^{j^{*}-1})/\gamma}}=\frac{a}{2^{\gamma^{j^{*}-2}}}=a_{j^{*}-2}.
  \]
\end{proof}
\begin{lem}
  \label{lem:balanced enough}$A(V({\cal S}))\ge|A|/L^{2}$.
\end{lem}

\begin{proof}
  By the covering property of $\textsc{cluster}$ from \Cref{lemma:clustering}, and as the Algorithm invokes \Cref{lemma:clustering} with $R=2\Delta_{i^{*}}$, we have that
  \begin{align*}
    V({\cal S})=\bigcup_{S\in{\cal S}}S & \supseteq\bigcup_{v\in N}B(v,2\Delta_{i^{*}}) \supseteq V_{i^{*},j^{*}}
  \end{align*}
  where the final inclusion follows as $N$ is a maximal packing of balls of radius $\Delta_{i^*}$ centered at vertices in $V_{i^{*},j^{*}}$. The claim now follows because $A(V_{i^{*},j^{*}})\ge|A|/L^{2}$ by the choice of $(i^{*},j^{*})$.
\end{proof}
Now, we are ready to conclude the bound on $|C|$. This would complete the proof.
\begin{lem}
  Let $C=\text{ED}(G,A,\phi)$. We have $|C|\le\phi\beta(|A|)|A|\log|A|$ where $\beta(|A|)=c_{1}8^{L}L^{2}\gamma^{2}=\exp(O(\sqrt{\log\log|A|})$ and $c_{1}$ is some constant.
  \end{lem}

  \begin{proof}
    There are two cases.

    \paragraph{(Heavy-component case): }

    Assume w.l.o.g. that $A(S')\le|A|/2$. By \Cref{lem:heavy cut}, we have
    \begin{align*}
      |C| & \le12\phi A(S')+|\text{ED}(G[S'],A_{S'},\phi)|+|\text{ED}(G[V-S'],A_{V-S'},\phi)|\\
      & \le12\phi|A_{S'}|+\phi\beta(|A|)|A_{S'}|(\log|A|-1)+\phi\beta(|A|)|A_{V-S'}|\log|A|\\
      & \le\phi\beta(|A|)|A_{S'}|\log|A|+\phi\beta(|A|)|A_{V-S'}|\log|A|\\
      & =\phi\beta(|A|)|A|\log|A|.
    \end{align*}
    The second inequality is because $A(S')\le|A|/2$ and the third is because $12\le\beta(|A|)$.

    \paragraph{(Balanced case): }

    We have
    \begin{align*}
      |C| & \le c_{0}\phi8^{L}\gamma^{2}\cdot|A|\log\frac{|A|}{a_{j^{*}-2}}+\sum_{S\in{\cal S}}|\text{ED}(G[S],A_{S},\phi)|+|\text{ED}(G[V-V({\cal S})],A_{V-V({\cal S})},\phi)|\\
      & \le c_{0}\phi L^{2}8^{L}\gamma^{2}\cdot|A_{V({\cal S})}|\log\frac{|A|}{a_{j^{*}-2}}+\sum_{S\in{\cal S}}\phi\beta(|A_{S}|)|A_{S}|\log|A_{S}|+\phi\beta(|A_{V-V({\cal S})}|)|A_{V-V({\cal S})}|\log|A_{V-V({\cal S})}|\\
      & \le\phi\beta(|A|)|A_{V({\cal S})}|\log\frac{|A|}{a_{j^{*}-2}}+\phi\beta(|A|)|A_{V({\cal S})}|\log a_{j^{*}-2}+\phi\beta(|A|)|A_{V-V({\cal S})}|\log|A|\\
      & =\phi\beta(|A|)|A|\log|A|.
    \end{align*}
    The first line is by \Cref{lem:cluster cut}. The second line is by \Cref{lem:balanced enough} and the induction hypothesis. The third line is because $\beta(|A|)=c_{1}8^{L}L^{2}\gamma^{2}\ge c_{0}L^{2}8^{L}\gamma^{2}$ and \Cref{lem:bound size cluster}.
  \end{proof}

  \section{Heavy-Component Case}

  \label{sec:heavy case}
Throughout this section, $\mathrm{diam}_{\ell}(K):=\max_{u,v\in K}\dist_{\ell}(u,v)$ denotes the (strong) diameter of $K$ in the shortest-path metric induced by $\ell$.

  Observe that the lemma below implies \Cref{lem:heavy cut}.
  \begin{lem}
    Let $G=(V,E)$ be a graph with edge length $\ell$ and $A$ be a node-weighting. Suppose that $\sum_{e\in E}\ell_{e}\le1\text{ and }\sum_{u,v}D(u,v)\text{dist}_{\ell}(u,v)\ge1/\phi.$ Let $K\subseteq V$ be a set where $A(K)\ge A(V)/3$ and diameter $\mathrm{diam}_{\ell}(K)\le\frac{1}{4\phi|A|}$. Then, there is an efficient algorithm that finds a set $S'\supseteq K$ where
    \[
      |\delta(S')|\le12\phi\min\{A(S'),A(V-S')\}.
    \]
  \end{lem}

  \paragraph{Algorithm Sweep Cut.}

  Define $\pi(v):=\text{dist}_{\ell}(v,K)$ and sort vertices so that $\pi(v_{1})\ge\pi(v_{2})\ge\dots\ge\pi(v_{n}).$ Let $S_{k}=\{v_{1},\dots,v_{k}\}$ Return
  \[
    S'\gets\arg\min_{1\le k<n}\frac{|\delta(S_{k})|}{D(S_{k},V-S_{k})}
  \]
  where $D(S,T)=\sum_{s\in S,t\in T}D(s,t)$.


  \paragraph{Analysis.}

  Since $D$ respects $A$, we have
  \[
    \frac{|\delta(S_{k})|}{\min\{A(S_{k}),A(V-S_{k})\}}\le\frac{|\delta(S_{k})|}{D(S_{k},V-S_{k})}.
  \]
  So, it suffices to show
  \[
    \min_{1\le k<n}\frac{|\delta(S_{k})|}{D(S_{k},V-S_{k})}\le12\phi.
  \]
  We have
  \begin{align}
    \min_{1\le k<n}\frac{|\delta(S_{k})|}{D(S_{k},V-S_{k})} & =\min_{1\le k<n}\frac{|\delta(S_{k})|\cdot|\pi(v_{k})-\pi(v_{k+1})|}{D(S_{k},V-S_{k})\cdot|\pi(v_{k})-\pi(v_{k+1})|} \nonumber \\
    & \le\frac{\sum_{k<n}|\delta(S_{k})|\cdot|\pi(v_{k})-\pi(v_{k+1})|}{\sum_{k<n}D(S_{k},V-S_{k})\cdot|\pi(v_{k})-\pi(v_{k+1})|} \nonumber \\
    & =\frac{\sum_{(u,v)\in E}|\pi(u)-\pi(v)|}{\sum_{u,v}D(u,v)\cdot|\pi(u)-\pi(v)|}.
    \label{eq:ratio}
  \end{align}
  To see the last equality, for each $(u,v)=(v_{i},v_{j})$, its total contribution to $\sum_{k<n}|\delta(S_{k})|\cdot|\pi(v_{k})-\pi(v_{k+1})|$ is exactly
  \[
    |\pi(v_{i})-\pi(v_{i+1})|+\dots+|\pi(v_{j-1})-\pi(v_{j})|=|\pi(v_{i})-\pi(v_{j})|.
  \]
  For each demand pair $D(v_{i},v_{j})$, its total contribution to $\sum_{k<n}D(S_{k},V-S_{k})\cdot|\pi(v_{k})-\pi(v_{k+1})|$ is exactly
  \[
    D(v_{i},v_{j})|\pi(v_{i})-\pi(v_{i+1})|+\dots+D(v_{i},v_{j})|\pi(v_{j-1})-\pi(v_{j})|=D(v_{i},v_{j})|\pi(v_{i})-\pi(v_{j})|.
  \]

  \paragraph{Numerator is $\le1$.}

  We bound each term in the sum. For each edge $(u,v)$, we have
  \[
    |\pi(u)-\pi(v)|=|\text{dist}_{\ell}(u,K)-\text{dist}_{\ell}(v,K)|\le\text{dist}_{\ell}(u,v)\le\ell_{uv}
  \]
  So
  \[
    \sum_{e=(u,v)\in E}|\pi(u)-\pi(v)|\le\sum_{e\in E}\ell_{e}\le1.
  \]

  \paragraph{Denominator $\ge1/12\phi$.}

  Since $\pi(v)=0$ for all $v\in K$,
  \begin{align}
    \sum_{u,v}D(u,v)\cdot|\pi(u)-\pi(v)| & \ge\sum_{u}\pi(u)D(u,K)\ge\frac{1}{3}\sum_{u}\pi(u)A(u).
    \label{eq:demand-cut-distance}
  \end{align}
  where the last inequality is by the product structure of $D$: $D(u,K)=A(u)\frac{A(K)}{|A|}\ge\frac{1}{3}A(u).$ Next, we have
  \begin{align*}
    1/\phi & \le\sum_{u,v}D(u,v)\text{dist}_{\ell}(u,v)\\
    & \le\sum_{u,v}D(u,v)(\text{dist}_{\ell}(u,K)+\text{diam}_{\ell}(K)+\text{dist}_{\ell}(K,v))\\
    & \le2\sum_{u}\pi(u)A(u)+\sum_{u,v}D(u,v)\cdot\frac{1}{4\phi|A|}  \tag{$\text{diam}_{\ell}(K)\leq \frac{1}{4\phi|A|}, \sum_{v} D(u,v) \leq A(u)$}.
  \end{align*}
  As $\sum_{u,v}D(u,v)\cdot\frac{1}{4\phi|A|}=1/4\phi$, plugging this above and rearranging gives 
  \[
    \frac{1}{4\phi}\le\sum_{u}\pi(u)A(u).
  \]
Together with \eqref{eq:demand-cut-distance} this implies that denominator in \eqref{eq:ratio} is at least $1/12\phi$. Thus we have that  \[\min_{1\le k<n}\frac{|\delta(S_{k})|}{D(S_{k},V-S_{k})}\le12\phi,\] which completes the proof of the heavy-component case.

\section{Open Questions}

\paragraph{All-or-nothing Flow: From $\log^{2}k$ to $\log^{1+o(1)}k$?}

In the \emph{all-or-nothing flow} problem~\cite{chekuri2013all}, we are given an undirected graph
$G=(V,E)$ with unit edge capacities and a set of $k$ demand pairs
$P=\{(s_1,t_1),\dots,(s_k,t_k)\}$.
The goal is to choose a subset $P'\subseteq P$ and, for each $(s_i,t_i)\in P'$, route one unit of
(splittable) flow from $s_i$ to $t_i$ so that the total load on every edge is at most~$1$.
Equivalently, we seek a maximum-cardinality subset of pairs that can be routed simultaneously with
congestion~$1$.

This problem can be viewed as a relaxation of the \emph{edge-disjoint paths} (EDP) problem, where each selected commodity must be routed on a \emph{single path}.
All-or-nothing flow was introduced as a ``nice'' intermediate model that admits polylogarithmic
approximation guarantees, whereas EDP with congestion~$1$ appears substantially harder even in very restricted graph classes \cite{chuzhoy2018almost}.

The current best approximation ratio for all-or-nothing flow in general graphs is
$O(\log^2 k)$~\cite{chekuri2005multicommodity}.
A key bottleneck is the $O(\log^2 k)$ \emph{overhead} incurred by the \emph{well-linked decomposition} for all-or-nothing flow. 
Informally, it repeatedly finds sparse cuts and recurses (in the spirit of expander decomposition),
and the loss can be viewed as the product of
(i) an $O(\log k)$ flow--cut gap (when there are $\Theta(k)$ terminals) and
(ii) an $O(\log k)$ recursion depth.

Our new flow-expander decomposition suggests that one might be able to \emph{telescope} the loss across
levels of recursion and reduce the total overhead to $\log^{1+o(1)}k$. However, the connection is not
black-box: the well-linked decomposition used for all-or-nothing flow interleaves
\emph{concurrent} multicommodity flow computations with \emph{maximum-throughput} multicommodity flow
steps, whereas flow-expander decomposition mostly only needs concurrent-flow computation at each recursive call.

Can we construct the well-linked decomposition of~\cite{chekuri2005multicommodity} with
$\log^{1+o(1)}k$ overhead using our technique? If so, it would immediately imply the following.

\begin{conjecture}\label{conj:aonf}
There is a polynomial-time $(\log^{1+o(1)}k)$-approximation algorithm for the all-or-nothing flow
problem.
\end{conjecture}

\paragraph{Tree sparsifiers.}

\emph{Tree cut/flow sparsifiers} can be viewed as a hierarchical analogue of expander decomposition and
underlie many routing and graph-optimization applications, including oblivious routing~\cite{racke2002minimizing},
online multicut~\cite{alon2006general}, near-linear-time approximate maximum flow~\cite{racke2014computing,peng2016approximate},
almost-linear-time minimum-cost flow~\cite{van2024almost}, and dynamic graph algorithms for connectivity and more~\cite{goranci2021expander}.

For disjoint sets $S,T\subseteq V$, let $\mincut_G(S,T)$ be the value of a minimum cut separating
$S$ from $T$ in $G$.
A \emph{tree cut sparsifier} for $G$ with quality $\gamma$ is a tree $T$ (with edge capacities and
$V(T)\supseteq V$) such that for all disjoint $S,T\subseteq V$,
\[
  \mincut_G(S,T)\ \le\ \mincut_T(S,T)\ \le\ \gamma\cdot \mincut_G(S,T).
\]
A \emph{tree flow sparsifier} for $G$ with quality $\gamma$ is a tree $T$ such that for every
$\deg_G$-respecting demand matrix $D$ (i.e., $\sum_u D(v,u)\le \deg_G(v)$ for all $v$),
\begin{itemize}
\item if $D$ is routable in $G$ with congestion~$1$, then $D$ is routable in $T$ with congestion~$1$, and
\item if $D$ is routable in $T$ with congestion~$1$, then $D$ is routable in $G$ with congestion at most $\gamma$.
\end{itemize}
As with cut vs.\ flow expansion, tree flow sparsifiers are strictly stronger: every tree flow
sparsifier of quality $\gamma$ is also a tree cut sparsifier of quality $\gamma$, whereas a tree cut
sparsifier of quality $\gamma$ only implies a tree flow sparsifier with an additional $O(\log n)$ loss
in general.

The qualitative picture here mirrors expander decomposition.
There is an $\Omega(\log n)$ graph-theoretic lower bound (already for tree cut sparsifiers, e.g.\ on
grids)~\cite{maggs1997exploiting,bartal1997line}.
Existentially, this barrier can be matched up to lower-order factors: tree cut sparsifiers of quality
$O(\log n\log\log n)$ exist for every graph~\cite{racke2014improved}.

However, all known \emph{polynomial-time} constructions lose an additional polylogarithmic factor
beyond the existential bound: the best current algorithms achieve roughly
$O(\log^{1.5} n\log\log n)$ for tree cut sparsifiers~\cite{racke2014improved} and
$O(\log^{2} n\log\log n)$ for tree flow sparsifiers~\cite{harrelson2003polynomial}.
Recent work \cite{racke2014computing,henzinger2025improved,agassy2025improved} focused on faster construction time but did \emph{not} lead to improved quality.

Given that our work essentially removes the analogous extra loss for flow-expander decompositions,
it is natural to ask whether the same is possible for tree flow sparsifiers.

\begin{conjecture}\label{conj:tree-flow}
There is a polynomial-time algorithm that constructs a tree flow sparsifier with quality
$\log^{1+o(1)} n$ for every $n$-vertex undirected graph.
\end{conjecture}

\paragraph{Vertex and directed expander decompositions.}

Our results concern edge-based decompositions in undirected graphs. A very natural next step is to ask whether the
same near-optimal overhead guarantees extend to (i) \emph{vertex}-expander decompositions as defined in \cite{long2022near} and (ii) \emph{directed}-expander decompositions as defined in \cite{bernstein2020deterministic,fleischmann2025improved}.
Can we obtain $\log^{1+o(1)} n$-overhead analogues in these settings as well?

  \bibliographystyle{alpha}
  \bibliography{ref}

\appendix

\section{Proof of \Cref{fact:product-demand}}
\label{sec:omit}

  Let $D'$ be any $A$-respecting demand. Consider routing $D'$ in the following complete graph $G_A$ with capacities given by $D_{A}$ as follows: for each unordered pair $\{x,y\}$ and each intermediate vertex $z\in V$, send an amount $D'(x,y)\cdot \frac{A(z)}{|A|}$ from $x$ to $z$ and an equal amount from $z$ to $y$. This defines a feasible (fractional) two-hop routing.

  Fix an (undirected) edge $\{x,z\}$ in $G_A$. The total flow sent on $\{x,z\}$ is at most
  \[
    \sum_{y\in V\setminus\{x\}} D'(x,y)\cdot\frac{A(z)}{|A|} \le A(x)\cdot\frac{A(z)}{|A|} = D_{A}(x,z),
  \]
  and the same bound holds for the load contributed when $x$ appears as the second endpoint. Thus every edge $\{u,v\}$ is used with congestion at most $2$ compared to its capacity $D_{A}(u,v)$.

  Therefore, $D'$ is routable in $G_A$ with congestion $2$. Composing this routing with a congestion-$1/\phi$ routing of $D_{A}$ in $G$ yields a routing of $D'$ in $G$ with congestion $2/\phi$, proving that $A$ is $(\phi/2)$-flow-expanding in $G$.

  \end{document}